\documentclass[nofootinbib,onecolumn ,aps,amsmath,amssymb,longbibliography,a4paper,superscriptaddress,tightenlines,notitlepage,12pt]{revtex4-1}

\usepackage{graphicx}
\usepackage{amsmath}
\usepackage{amssymb}
\usepackage{amsfonts}
\usepackage{amsthm}
\usepackage{mathtools}
\usepackage{hyperref}
\usepackage{xcolor}
\usepackage[T1]{fontenc}
\usepackage[english]{babel}
\usepackage{braket}
\usepackage{enumitem}
\usepackage{tikz}
\usepackage{bbm}

\usepackage{times}

\hypersetup{colorlinks=true,urlcolor=[rgb]{0,0,0.5},citecolor=[rgb]{0.5,0,0},linkcolor=[rgb]{0,0,0.4}}

\usepackage[bbgreekl]{mathbbol}
\DeclareSymbolFontAlphabet{\mathbb}{AMSb}
\DeclareSymbolFontAlphabet{\mathbbl}{bbold}

\newtheorem{theorem}{Theorem}

\newtheorem{lemma}{Lemma}

\newtheorem{definition}{Definition}
\newtheorem{conjecture}{Conjecture}

\newtheorem{remark}{Remark}

\def\autorefapp#1{\hyperref[#1]{Appendix~\ref{#1}}}

\begin{document}

\title{On the moments of random quantum circuits\\ and robust quantum complexity}

\author{Jonas Haferkamp}
\affiliation{School of Engineering and Applied Sciences,
	Harvard University, Cambridge, MA 02318, USA}

\begin{abstract}\noindent
	We prove new lower bounds on the growth of robust quantum circuit complexity -- the minimal number of gates $C_{\delta}(U)$ to approximate a unitary $U$ up to an error of $\delta$ in operator norm distance.
 More precisely we show two bounds for random quantum circuits with local gates drawn from a subgroup of $SU(4)$.
	First, for $\delta=\Theta(2^{-n})$, we prove a linear growth rate: $C_{\delta}\geq d/\mathrm{poly}(n)$ for random quantum circuits on $n$ qubits with $d\leq 2^{n/2}$ gates.
	Second, for $ \delta=\Omega(1)$, we prove a square-root growth of complexity: $C_{\delta}\geq \sqrt{d}/\mathrm{poly}(n)$ for all $d\leq 2^{n/2}$.
    Finally, we provide a simple conjecture regarding the Fourier support of randomly drawn Boolean functions that would imply linear growth for constant $\delta$.
	While these results follow from bounds on the moments of random quantum circuits, we do not make use of existing results on the generation of unitary $t$-designs.
	Instead, we bound the moments of an auxiliary random walk on the diagonal unitaries acting on phase states.
 In particular, our proof is comparably short and self-contained.
\end{abstract}

\maketitle

Applications of quantum circuit complexity range from topological phases of matter to the theory of black holes.
In the context of the AdS/CFT correspondence, Brown and Susskind conjectured that complexity grows linearly for an exponentially long time in the dynamic of generic quantum systems~\cite{brown2018second}.
While proving lower bounds on circuit complexity is notoriously hopeless for individual states and unitaries, concrete progress can be made for random ensembles.
In particular, random quantum circuits provide a powerful model for the dynamic of disordered systems.

A key feature of random quantum circuits is the quick generation of unitary designs as proven in the seminal paper Ref.~\cite{brandao_local_2016}.
The result in Ref.~\cite{brandao_local_2016} can be used to imply a lower bound on the quantum circuit complexity of random quantum circuits of depth $d$ that scales like $\Omega(d^{1/10.5})$~\cite{brandao_complexity_2019,roberts2017chaos}.
Subsequently, this bound was improved to $\Omega(d^{1/(5+o(1))})$ in Ref.~\cite{haferkamp2022random}.
In the limit of large local dimensions, a linear growth rate was proven in Ref.~\cite{hunter2019unitary}.
Similarly, a (near) linear growth of complexity can be proven as long as the local dimension $q$ satisfies $q\geq 6t^2$ in Ref.~\cite{haferkamp2020improved}.
Ref.~\cite{jian2022linear} provides ideas based on path integrals towards proving linear growth in Brownian dynamics.
Lastly, the late time regime of random quantum circuits was studied in Ref.~\cite{oszmaniec2022saturation}, proving that saturation and recurrences appear after an exponential and doubly exponential time respectively.

For random quantum circuits over qubits, linear growth of complexity was established in Ref.~\cite{haferkamp2021linear} (see also~\cite{li2022short}).
However, the methods of Ref.~\cite{haferkamp2021linear} apply to a particularly brittle notion of circuit complexity, which does not account for a quantifiable implementation error $\delta>0$.
In particular, the gates involved in Refs.~\cite{haferkamp2021linear,li2022short} cannot necessarily be represented with a bounded number of bits.

In this work, we prove two lower bounds on the circuit complexity for random quantum circuits with gates drawn from a specific subgroup of $SU(4)$.
First, we prove a linear lower bound $\Omega(d/\mathrm{poly}(n))$ for all $d\leq 2^{n/2}$ on the minimal number of gates required to approximate a unitary up to errors $\delta=\Theta(2^{-n})$.
In particular, this implementation error allows for a representation with linearly many bits.
Therefore, our result cannot be explained as an artefact of working with a continuous parameter space as opposed to the more combinatorial notion of complexity for Boolean functions.
For constant implementation errors we show a growth of $\Omega(\sqrt{d}/\mathrm{poly}(n))$ for the circuit complexity.
Moreover, we formulate Conjecture~\ref{conjecture:boolean} on the Fourier support of randomly drawn boolean functions, which would directly imply a version of the Brown-Susskind conjecture for random quantum circuits with robustness $\delta=\Omega(1)$.
While we do not prove generation of unitary $t$-designs, both of our results follow from moment bounds.
This work is comparably self-contained and only uses well-known bounds on the mixing of random walks of finite groups~\cite{diaconis1993comparison} and Parseval's identity~\cite{o2014analysis} for the Fourier transform of boolean functions.
In particular, we do not require techniques to prove spectral gaps such as Ref.~\cite{nachtergaele1996spectral} or ~\cite{knabe1988energy} nor the path coupling technique on the unitary group~\cite{oliveira2009convergence}.

The subgroup of $SU(4)$ we draw gates from is generated by CNOTs and diagonal unitaries.
We can prove that these effectively generate a quickly mixing random walk on phase states.
Random phase states are also used for constructions of computational quantum pseudorandomness, i.e. ensembles of states that are indistinguishable from the Haar measure with efficient algorithms~\cite{ji2018pseudorandom,bouland2022quantum}.

\section{Preliminaries}
For any compact Lie group $G$, we denote by $\mu_G$ the uniform (left-Haar) measure on $G$.
For a probability measure $\nu$ on $SU(2^k)$ with $k\leq n$, denote by $\nu_{i_1,\ldots,i_k}$ the embedding in $SU(2^n)$ on the tensor product of the qubits $i_1,\ldots,i_k$.
The trace norm (or Schatten $1$-norm) of a matrix $A$ is denoted by $||A||_1:=\mathrm{Tr}[|A|]$, where $|A|=\sqrt{AA^{\dagger}}$ and its operator norm by $||A||_{\infty}$.

We define the following standard model of random quantum circuits:
\begin{definition}[Local random quantum circuits]\label{definition:randomquantumcircuits}
Random quantum circuits are defined as a random walk on the unitary group $SU(2^n)$: Draw a Haar random unitary from $SU(4)$ and apply them to a randomly drawn pair $(i,i+1)$ of nearest neighbor qubits.
Here, we impose periodic boundary conditions, i.e. we identify qubit $n+1$ and $1$.
We also define $\nu$-random quantum circuits, where we draw the random gates from some probability measure $\nu$ on $SU(4)$.
We denote the corresponding probability measure on $SU(2^n)$ by $\nu_{(n)}$.
\end{definition}
\begin{definition}
We define a probability measure $\zeta$ on $SU(4)$:
With probability $1/2$: Draw a gate from the group generated by $\mathrm{CNOT}$s on $2$ qubits.
This is a finite group of $6$ elements that we denote by $\mathcal{Z}_2$.
Also with probability $1/2$: Draw a gate $e^{\mathrm{i}\phi Z}\otimes \mathbbm{1}_2$ with $\phi$ uniformly distributed over $[0,2\pi]$.
\end{definition}
\begin{remark}
    We remark that it is not essential how the gates are drawn from the subgroup $\langle \mathrm{CNOT}_{1,2}, \mathrm{CNOT}_{2,1}\rangle\times \mathrm{Diag}_4\subset SU(4)$, where $\times$ denotes the element-wise multiplication and $\mathrm{Diag}_4$ is the subgroup of diagonal unitaries.
    In particular, drawing uniformly from this group would also yield the desired results.
\end{remark}

We use of the concept of moment operators of probability measures on the unitary group.
These encode how the $t$-moments of the measure behave.
More precisely, the matrix entries of
\begin{equation}
M(\nu,t):=\mathbb{E}_{U\sim \nu} U^{\otimes t}\otimes \overline{U}^{\otimes t}
\end{equation} 
are the expectation values of balanced monomials of degree $t$.
Here, $\overline{()}$ denotes the entry-wise complex conjugate.
We will repeatedly use the notation $U^{\otimes t,t}:= U^{\otimes t}\otimes \overline{U}^{\otimes t}$ and $|\psi\rangle^{\otimes t,t}:=|\psi\rangle^{\otimes t}\otimes \overline{|\psi\rangle}^{\otimes t}$.
These operators are hermitian when $\nu$ is a symmetric measure, i.e. all expectation values are invariant under $()^{\dagger}$.
The highest eigenvalue of $M(\nu,t)$ is always $1$ and we will refer to the gap between the $1$ and the second highest eigenvalue as the spectral gap of $M(\nu,t)$.
\section{Results}
Here, we present the two main results regarding the growth of complexity in random quantum circuits.
We define robust quantum circuit complexity for states and unitaries:
\begin{definition}[Quantum circuit complexity]
	Let $|\psi\rangle$ be a state on $n$ qubits and $\mathcal{G}\subseteq SU(4)$.
	Then, for $0\leq \delta<1$, we defined $C_{\mathcal{G},\delta}(|\psi\rangle)$ to be the minimal number $m$ of gates $V_1,\ldots,V_m$ such that $|\phi\rangle=V_m\cdots V_1|0^n\rangle$ satisfies 
	\begin{equation}
\frac12	|| |\psi\rangle\langle\psi|-|\phi\rangle\langle \phi| ||_1\leq \delta.
	\end{equation}
	For a unitary $U$, we define its quantum circuit complexity $C_{\mathcal{G},\delta}(U)$ as the minimal number $m$ of gates $V_1,\ldots,V_m$ such that 
	\begin{equation}
	||U-V_1\cdots V_m||_{\infty}\leq \delta.
	\end{equation}
		We further set $C_{\delta}:=C_{SU(4),\delta}$.
\end{definition}
Our main result is the following theorem:
\begin{theorem}[Growth of complexity in random quantum circuits]\label{thm:square-root-growth}
	Consider a unitary $U$ generated by $\zeta$-random quantum circuits with $d$ $d\leq 2^{n/2}$ gates.
	Then, for any $\delta\geq 0$ with $\delta<\sqrt{1-2^{-1/2}}$, we have with probability $1-\exp(-\Omega(d/n^8))$ that 
	\begin{equation}\label{eq:linearboundintheorem}
	C_{\tilde{\delta}}(U)\geq \frac{d}{K n^{9}\log(n)\log_2(1/\delta)}\left(1-2\log_2\left(\frac{1}{1-\delta^2}\right)\right),
	\end{equation}
	where 
	\begin{equation}
	\tilde{\delta}:=\frac{(1-\sqrt{1-\delta^2})d}{40000n^8 2^n}
	\end{equation}
	and $K>0$ a constant.
	Moreover, with probability $1-\exp(-\Omega(\sqrt{d}/n^{4}))$ and
for $\delta\in (0,1/\sqrt{2})$ it holds that
		\begin{equation}\label{eq:complexitybound}
	C_{\delta}(U|0^n\rangle)\geq \frac{\sqrt{d}}{ Kn^{4}\log_2(d/\delta^2)} \left(n-\log_2\left(\frac{\sqrt{d}}{200n^{4}}\right)-2\log_2\left(\frac{1}{1-2\delta^2}\right)\right).
	\end{equation}
\end{theorem}
We remark that the constant $K$ can be explicitly upper bounded using e.g. techniques from Ref.~\cite{harrow2002efficient}.
The scaling in Eq.~\eqref{eq:linearboundintheorem} will follow from a moment bound at late times $d'=\mathrm{poly}(n)2^n$.
At this depth, the moment bound in Theorem~\ref{prop:momentbound} implies that the robust circuit complexity is maximal $\Omega(2^n)$ with probability $1-O(2^{-2^{n}})$.
We will then split the deep random quantum circuit into blocks with $d$ gates that are identically distributed and observe that a complexity of $o(d/\mathrm{poly}(n))$ with probability $\omega(2^{-d/\mathrm{poly}(n)}))$ leads to a contradiction:
All these instances add up to circuits with complexity $o(2^n)$ with probability at least $\omega(2^{-2^n})$.
Unfortunately, by adding the circuits we loose the robustness of the statement as errors can accumulate. 
The details of this argument can be found in Section~\ref{section:growthofcomplexity}.

\section{Moment bounds}
In this section we prove the following bounds on the moments of random quantum circuits. 
We will then apply these results in the subsequent section to obtain the main results presented in the previous section.
\begin{theorem}\label{prop:momentbound}
	We have the following bound on the moments of $\zeta$-random quantum circuits of depth $d$ and any state $|\psi\rangle$
	\begin{equation}\label{eq:momentboundintheorem}
	\mathbb{E}_{\zeta_{(n)}^{*d}}|\langle \psi|U|+^n\rangle|^{2t}\leq \sqrt{\frac{t!}{2^{nt}}+ e^{-d/40000n^7}+\left(1-\frac{1}{2t}+3\times 2^{-n}\right)^{d/16000n^7}}.
	\end{equation}
	Moreover,
		 \begin{equation}\label{eq:finalmomentboundintheorem}
	\mathbb{E}_{U\sim \zeta_{(n)}^{*d}} |\langle\psi|U|+^n\rangle|^{2t}\leq \sqrt{t!2^{2^{n}/2}2^{-nt}+ e^{- d/40000n^7}+\left(1-\frac{1}{2(n+1)}+3\times 2^{-n}\right)^{d/16000n^7}}.
	\end{equation}
\end{theorem}
The first bound yields strong upper bounds for all $t$, but requires quadratic depth in $t$.
The second bound only produces reasonable bounds for $t=\Omega(2^n/\mathrm{poly}(n))$, but for a depth linear in $t$.
We remark here that the first bound Eq.~\eqref{eq:momentboundintheorem} even implies the square-root growth of  stronger notion of quantum circuit complexity that rules out the distinguishability from the maximally mixed stated with small circuits as proven in Ref.~\cite{brandao_complexity_2019,brandao2016efficient}.

In this section we present the proof of Theorem~\ref{prop:momentbound}.

\begin{proof}[Proof of Theorem~\ref{prop:momentbound}]

We first derive a bound on the moments
\begin{equation}
\mathbb{E}_{U\sim \zeta_{(n)}^{*d}}|\langle+^n|U|+^n\rangle|^{2t}=\mathbb{E}_{U_1\sim \zeta_{(n)}}\cdots \mathbb{E}_{U_d\sim \zeta_{(n)}}|\langle+^n|U_1\cdots U_d|+^n\rangle|^{2t}.
\end{equation}
We first sort the wheat (phases) from the chaff ($\mathrm{CNOT}$s): Commute every gate $e^{\mathrm{i}\phi Z}$ through the circuit to the ket $|+^n\rangle$.
Here, it will act non-trivially as $e^{\mathrm{i}\zeta_{(n)j} Z^y}$ with $Z^y:=Z^{y[1]}\otimes \cdots \otimes Z^{y[n]}$ for a random bitstring $y\in \{0,1\}^n$.
Of course, both the number of rotations and the bitstrings $y$ depend on the $\mathrm{CNOT}$s that where drawn.
 At the same time, the $\mathrm{CNOT}$s act on the bra $\langle +^n|$ trivially.
Denote by $p_{y}$ the monomial defined by $p_y(x)=x[1]^{y[1]}\cdots x[n]^{y[n]}$.
We find
\begin{align}
\begin{split}\label{eq:momentcal1}
\mathbb{E}_{U\sim \zeta_{(n)}^{*d}}|\langle+^n|U|+^n\rangle|^{2t}&=\mathbb{E}_{m,y_1,\ldots,y_m, \phi_1,\ldots,\phi_m}|\langle+^n|\prod_j e^{\mathrm{i}\phi_j Z^{y_j}}|+^n\rangle|^{2t}\\
&=\mathbb{E} \frac{1}{2^{2nt}}\left(\sum_{x}e^{\mathrm{i}\sum_{j}\phi_jp_{y_j}(x)}\right)^t\left(\sum_{x'}e^{-\mathrm{i}\sum_{j}\phi_jp_{y_j}(x)}\right)^t\\
&=\frac{1}{2^{2nt}} \mathbb{E}\sum_{x_1,\ldots,x_t,x'_1,\ldots,x'_t}e^{\mathrm{i}\sum_{j}\phi_j\sum_{l=1}^{t}\left(p_{y_j}(x_l)-p_{y_j}(x'_l)\right)}.
\end{split}
\end{align}
We can split this sum in two cases: all pairs of tuples $(x_1,\ldots,x_t),(x'_1,\ldots,x'_t)$ that can be obtained from each other via a permutation of the $t$ bitstrings and the rest.
We denote equivalence modulo permutations by $\sim$ and find:
\begin{align}
\begin{split}\label{eq:momentcal2}
\eqref{eq:momentcal1}&=\frac{1}{2^{2nt}} \mathbb{E}\left(\sum_{(x_1,\ldots,x_t)\sim(x'_1,\ldots,x'_t)}e^{\mathrm{i}\sum_{j}\phi_j\sum_{l=1}^{t}\left(p_{y_j}(x_l)-p_{y_j}(x'_l)\right)}+\sum_{(x_1,\ldots,x_t)\not\sim(x'_1,\ldots,x'_t)}e^{\mathrm{i}\sum_{j}\phi_j\sum_{l=1}^{t}\left(p_{y_j}(x_l)-p_{y_j}(x'_l)\right)}\right)\\
&=  \frac{1}{2^{2nt}}\sum_{x_1,\ldots,x_t} |\{(x'_1,\ldots,x'_t)\sim (x_1,\ldots,x_t)\}|+ \frac{1}{2^{2nt}}\mathbb{E}\sum_{(x_1,\ldots,x_t)\not\sim(x'_1,\ldots,x'_t)}e^{\mathrm{i}\sum_{j}\phi_j\sum_{l=1}^{t}\left(p_{y_j}(x_l)-p_{y_j}(x'_l)\right)}\\
&\leq \frac{t!}{2^{nt}}+\frac{1}{2^{2nt}}\sum_{(x_1,\ldots,x_t)\not\sim(x'_1,\ldots,x'_t)}\mathbb{E} e^{\mathrm{i}\sum_{j}\phi_j\sum_{l=1}^{t}\left(p_{y_j}(x_l)-p_{y_j}(x'_l)\right)}.
\end{split}
\end{align}
In the limit of infinite depth we expect the second summand to vanish as our random walk on phase states converges to the uniform measure on $(S^1)^{2^n}$.

Notice that the contribution of a pair $(x_1,\ldots,x_t),(x'_1,\ldots,x'_t)$ averages to $0$ if any of the expressions $\sum_{l=1}^{t}\left(p_{y_j}(x_l)-p_{y_j}(x'_l)\right)$ do not vanish.
In the following we will say that $p_y$ \textit{distinguishes} $(x_1,\ldots,x_t) $ and $(x'_1,\ldots,x'_t)$ if $\sum_{l=1}^{t}\left(p_{y}(x_l)-p_{y}(x'_l)\right)\neq 0$.
We will apply Fourier analysis of Boolean functions~\cite{o2014analysis} to lower bound the probability that $\sum_{l=1}^{t}\left(p_{y_j}(x_l)-p_{y_j}(x'_l)\right)\neq 0$.
For this define the functions 
\begin{equation}\label{eq:definefunctions}
f_{x_1,\ldots,x_t}^{x'_1,\ldots,x'_t}:=\sum_{l=1}^t\delta_{x_l}-\delta_{x'_l}.
\end{equation}
If $(x_1,\ldots,x_t)\not\sim(x'_1,\ldots,x'_t)$ is equivalent to $f_{x_1,\ldots,x_t}^{x'_1,\ldots,x'_t}\neq 0$.
For a general function $f:\{-1,1\}^n\to \mathbb{R}$, the Fourier transform is defined as~\cite{o2014analysis}:
\begin{equation}
\hat{f}(y)=\mathbb{E}_x p_y(x)f(x).
\end{equation} 
It follows directly that
\begin{equation}
\hat{f}_{x_1,\ldots,x_t}^{x'_1,\ldots,x'_t}(y)=2^{-n}\sum_{l=1}^{t}\left(p_{y}(x_l)-p_{y}(x'_l)\right).
\end{equation}

We will first analyse the expression in Eq.~\eqref{eq:momentcal2} for the case that $m$ is fixed and $y_1,\ldots,y_m$ are distributed uniformly and i.i.d. at random.
We call the resulting random walk on the diagonal unitaries the \textit{ideal auxiliary walk} as depicted in Figure~\ref{figure:idealwalk}.
\begin{figure}
	\includegraphics[scale=1.2]{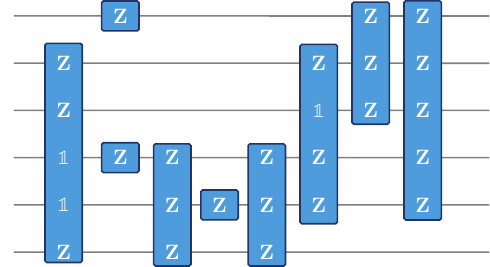}
	\caption{\label{figure:idealwalk} An instance of the ideal auxiliary walk. Each gate corresponds to a uniformly random rotation around the Pauli $Z$ string it is labeled with. 
		E.g. the first unitary in this figure is a roation $e^{\mathrm{i}\varphi \mathbbm{1}\otimes Z\otimes Z\otimes \mathbbm{1}\otimes \mathbbm{1}\otimes Z}$ with a uniformly random $\varphi\in (0,2\pi]$.}
\end{figure}
In the following we find two upper bounds that we can use to prove Eq.~\eqref{eq:momentboundintheorem}.
The probabilities $2^{-n}|\{y\in\{0,1\}^n,\hat{f}_{x_1,\ldots,x_t}^{x'_1,\ldots,x'_t}(y)= 0 \}|$ are the eigenvalues of the $t$-fold moment operator for the ideal auxiliary walk on the diagonals.
The first bound will be a bound on all non-trivial eigenvalues and therefore on its spectral gap. 
The second bound is stronger but only holds for a quantifiable number of the non-trivial eigenvalues.
Afterwards, we come back to the actual random walk $\zeta_{(n)}$ and prove that sufficiently nearly uniformly drawn $y_j$ are generated for the two bounds to be applicable.

\textbf{Bound on gap of the auxiliary walk.} 
For the first bound we apply Parseval's identity~\cite{o2014analysis}:
If a string appears in both $(x_1,\ldots,x_t)$ and $(x'_1,\ldots,x'_t)$ we can remove it without changing $f_{x_1,\ldots,x_t}^{x'_1,\ldots,x'_t}$.
We remove strings like this until we find two tuples of length $0<r\leq t$ with no common strings.
Then, we obtain from Parseval's identity~\cite{o2014analysis} that 
\begin{align}
\begin{split}
2^{-2n}(2r)^2 \left|\left\{y, \sum_{l=1}^{t}p_{y}(x_l)-p_{y}(x'_l)\neq 0\right\}\right|&\geq 2^{-2n}\sum_y \left(\sum_{l=1}^{r}p_{y}(x_l)-p_{y}(x'_l)\right)^2\\
&=\sum_y\left(\hat{f}_{x_1,\ldots,x_r}^{x'_1,\ldots,x'_r}(y)\right)^2\\
&= 2^{-n}\sum_x\left(f_{x_1,\ldots,x_r}^{x'_1,\ldots,x'_r}(x)\right)^2\\
&\geq 2r2^{-n}.
\end{split}
\end{align}
Thus, the probability of a random string $y$ distinguishing the two tuples is lower bounded by
\begin{equation}\label{eq:parsevalapplied}
2^{-n}\left|\left\{y, \sum_{l=1}^{t}p_{y}(x_l)-p_{y}(x'_l)\neq 0\right\}\right|\geq \frac{1}{2r}\geq \frac{1}{2t}.
\end{equation}
Therefore, drawing $m$ random bitstrings, the probability of not finding a single string with this property becomes $\leq (1-\frac{1}{2t})^m$.
By choosing $m\sim nt^2$, this sufficiently supresses the second summand in Eq.~\eqref{eq:momentcal1}.

\textbf{Counting argument for most eigenvalues.}
For the second bound, we use a counting argument to show that a large enough fraction of functions $f_{x_1,\ldots,x_t}^{x'_1,\ldots,x'_t}$ for $t=2^n/2$ has large support.
As $\{p_y\}_y$ is a basis~\cite{o2014analysis}, every function $f:\{-1,1\}^n\to \mathbb{R}$ has a unique decomposition in terms of monomials.
We find that the values of $2^n\hat{f}_{x_1,\ldots,x_t}^{x'_1,\ldots,x'_t}$ are integers between $-2^n$ and $2^n$.
If the Fourier support is smaller than $A$, this constitutes at most $2^{(n+1)A}$ distinct functions.
However, different pairs of tuples can yield the same corresponding function $f_{x_1,\ldots,x_t}^{x'_1,\ldots,x'_t}$.
We upper bound the maximal multiplicities that can arise in the second summand of Eq.~\eqref{eq:momentcal2} that way:
For any pair of tuples $(x_1,\ldots,x_t)$ and $(x'_1,\ldots,x'_t)$, we can remove $t-r$ bitstrings that appear in both tuples without changing the corresponding function.
Therefore, for each function $f_{x_1,\ldots,x_t}^{x'_1,\ldots,x'_t}$ there are at most $2^{n(t-r)}{t\choose t-r}t!/r!$ such ``invisible'' pairs of $t-r$ tuples.
Indeed, we choose $t-r$ positions in the first tuple leading to ${t\choose t-r}$ choices and for each position we have $2^{n}$ choices.
Moreover, for each of these constellations in the first tuple, we can enlist the $t-r$ bitstrings chosen and assign them to the second tuple choosing $t-r$ positions, where the order matters. This leads to at most $t!/r!$ further choices per constellation in the first tuple.
It is easy to see that the corresponding functions are now invariant precisely under permutations of the elements of the two remaining $r$-tuples, which yields $r!^2$ many choices.
In summary, the maximal multiplicity arising in Eq.~\eqref{eq:momentcal2} is upper bounded by $\max_r 2^{n(t-r)}{t\choose t-r}r!t!=\max_r 2^{n(t-r)}t!^2/(t-r)!$.
Therefore, the fraction of tuples with Fourier support smaller than $A$ can be upper bounded by  
\begin{equation}
t! 2^{(n+1)A}2^{-nt} \max_r(2^{-nr}t!/(t-r)!)\leq t! 2^{(n+1)A}2^{-nt}
\end{equation}
for $r\leq t\leq 2^n$, where we used that 
\begin{equation}
t!/(t-r)!\leq t^{r}=2^{r\log_2(t)}
\end{equation}
and $\log_2(t)\leq n$.
We will upper bound all contributions in Eq.~\eqref{eq:momentcal2} of tuples with Fourier support smaller than $A$ by $1$.
For all the remaining pairs of tuples the probability of not drawing a string that distinguishes them is upper bounded by $(1-A/2^n)$.

\textbf{Leaving the ideal setting.}
 In the following we show that the walk $\zeta_{(n)}$ generates sufficiently many nearly uniformly and nearly independently distributed phases $e^{\mathrm{i}\phi_j\sum_{l=1}^{t}\left(p_{y_j}(x_l)-p_{y_j}(x'_l)\right)}$ for the two bounds on the eigenvalues of the ideal walk to be applicable.
The random $\mathrm{CNOT}$ gates generate the subgroup of the Clifford group that normalizes all Pauli strings only containing $Z$ and identities as tensor factors.
In particular, the generated group is the finite group $\mathcal{Z}_n$ of reversible circuits with at most $2^{O(n^2)}$ elements.
Moreover, every element is generated by at most $9n$ $\mathrm{CNOT}$ gates~\cite{maslov2018shorter,bravyi2021hadamard}.
This suffices to show that a random walk $\sigma$ defined by applying random $\mathrm{CNOT}$ gates is gapped via the comparison technique by Diaconis and Shashahani~\cite{diaconis1993comparison}. 
\begin{lemma}\label{lemma:convergenceCNOTs}
	The total variation distance between the random walk $\sigma$ and the uniform measure on $\mathcal{Z}$ satisfies
	\begin{equation}
	d_{TV}(\sigma^{*k},\mu_{\mathcal{Z}})\leq 2^{n^2+n}\left(1-\frac{1}{500n^5}\right)^k.
	\end{equation}
\end{lemma}
We provide a proof of this lemma in Section~\ref{appendix:lemma1}.
In particular, we find
\begin{equation}\label{eq:TVdistance}
d_{TV}(\sigma^{*2000n^7},\mu_{\mathcal{Z}})\leq 2^{-n}.
\end{equation}
In the following, we therefore choose $$k=2000 n^7$$ and obtain that all probabilities over $\sigma^{*k}$ and $\mu_{\mathcal{Z}}$ differ at most by $2^{-n}$.
Moreover, notice that $\mu_{\mathcal{Z}}$ mix the Pauli $Z$-strings uniformly: $U Z^{\otimes y} U^{\dagger}$ with $y\neq 0\cdots 0$ is a uniformly random element of $\{Z^{\otimes x}\}_{x\in \{0,1\}^n,x\neq 0\cdots 0}$.
This follows from the fact that $\mathcal{Z}$ acts transitively on this set.
Moreover, the total variation distance between the uniform distribution on $\{Z^{\otimes x}\}_{x\in \{0,1\}^n,x\neq 0\cdots 0}$ and the uniform distribution on $\{Z^{\otimes x}\}_{x\in \{0,1\}^n}$ is upper bounded by $2\times 2^{-n}$.
Overall, the probability of drawing a distinguishing $Z$-string from the uniform measure vs the ensemble $\{UZ\otimes \mathbbm{1}_{n-1}U^{\dagger}\}_{U\sim \sigma^{*k}}$ differs by at most $3\times 2^{-n}$.

We now separate the $d$ random unitaries drawn from $\zeta_{(n)}$ into $d/4k$ many blocks. 
By Hoeffding's inequality we find that each block of $4k$ gates contains at least one $Z$ rotation and $k$ random generators of the group $\mathcal{Z}$ generated by $\mathrm{CNOT}$ with probability at least $1-2^{-n-1}$.
Therefore, the probability for two adjacent blocks to contain a $Z$ rotation in the left block and at least $k$ $\mathrm{CNOT}$s /identities in the right block is at least $(1-2^{-n-1})^2\geq 1-2^{-n}$. 
We can again apply Hoeffding's inequality to show that at least half of the pairs of blocks have this property with probability at least $1-e^{-2 \frac{d}{8k}(\frac12-2^{-n})^2}\leq 1-e^{- \frac{d}{20k}}$.
 Therefore, we find with high probability that there are at least $d/8k$ many rotations $e^{i\phi_j Z^{y_j}}$ with nearly uniform and nearly independent $y_j$ in Eq.~\eqref{eq:momentcal1}. 
 More precisely, up to errors of $3\times 2^{-n}$ in total variation distance by Eq.~\eqref{eq:TVdistance}. 
We find for all tuples $(x_1,\ldots,x_t)$, $(x'_1,\ldots,x'_t)$ with Fourier support larger than $A$ that 
 \begin{align}
 \begin{split}\label{eq:boundonphases}
 \mathbb{E} e^{\mathrm{i}\sum_{j}\phi_j\sum_{l=1}^{t}\left(p_{y_j}(x_l)-p_{y_j}(x'_l)\right)}&\leq e^{- \frac{d}{20k}}+\left(1-\frac{A}{2^n}+3\times 2^{-n}\right)^{d/8k}.
 \end{split}
 \end{align}
Similarly, using Eq.~\eqref{eq:parsevalapplied}, we find  for all tuples $(x_1,\ldots, x_t)\not \sim (x_1',\ldots,x'_t)$ that 
 \begin{equation}\label{eq:boundonphaseswitht}
 \mathbb{E}e^{\mathrm{i}\sum_{j}\phi_j\sum_{l=1}^{t}\left(p_{y_j}(x_l)-p_{y_j}(x'_l)\right)}\leq e^{- \frac{d}{20k}}+\left(1-\frac{1}{2t}+3\times 2^{-n}\right)^{d/8k}.
 \end{equation}

 \textbf{Putting it together.} 
 With this in mind, we obtain from Cauchy-Schwarz that 
 \begin{align}
 \begin{split}
 \mathbb{E}_{U\sim \zeta_{(n)}^{*d}} |\langle\psi|U|+^n\rangle|^{2t}&=\langle \psi|^{\otimes t,t} M(\zeta_{(n)}^{*d})|+^n\rangle^{\otimes t,t}\\
 &\leq \sqrt{\langle +^n|^{\otimes t,t} M(\zeta_{(n)}^{*d})^{\dagger}M(\nu_n^{*d},t)|+^n\rangle^{\otimes t,t}}\\
 &=\sqrt{\mathbb{E}_{U\sim \nu^{*2d}_n} |\langle +^n|U|+^n\rangle|^{2t}}.
 \end{split}
 \end{align}
 Eq.~\eqref{eq:finalmomentboundintheorem} follows from this combined with Eq.~\eqref{eq:boundonphases} and~\eqref{eq:momentcal2} for $t=2^n/2$.
  Eq.~\eqref{eq:momentboundintheorem} follows from the above calculation combined with Eq.~\eqref{eq:boundonphaseswitht} and~\eqref{eq:momentcal2}.
 \end{proof}

\section{Growth of complexity}\label{section:growthofcomplexity}
 The lower bounds in Theorem~\ref{thm:square-root-growth} now follow from a counting argument involving the union bound and Markov's inequality.
 \begin{proof}[Proof of Theorem~\ref{thm:square-root-growth}] 
Notice that the trace norm difference of pure states satisfies
\begin{equation}
\frac12 || |\phi\rangle\langle\phi|-|\psi\rangle\langle\psi| ||_1=\sqrt{1-|\langle\phi|\psi\rangle|^2}.
\end{equation}
Therefore, 
\begin{equation}
\frac12 || |\phi\rangle\langle\phi|-|\psi\rangle\langle\psi| ||_1\leq \delta \iff |\langle\phi|\psi\rangle|^2\geq 1-\delta^2.
\end{equation}
 Denote by $M_{\mathcal{K},R}$ all the unitaries that can be generated from applications of $R$ many gates from the gate set $\mathcal{K}$.
 Then, we obtain from Markov's inequality
 \begin{align}
 \begin{split}\label{eq:highermarkovbound}
 \mathrm{Pr}_{\phi}\left[\exists |\psi\rangle\in M_{\mathcal{K},R}, |\langle\psi|\phi\rangle|^2\geq 1-\delta^2\right]&\leq \sum_{\psi\in M_{\mathcal{K},R}}\mathrm{Pr}_{\phi}[|\langle\psi|\phi\rangle|^2\geq 1-\delta^2]\\
 &=\sum_{\psi\in M_{\mathcal{K},R}}\mathrm{Pr}_{\phi}[|\langle\psi|\phi\rangle|^{2t}\geq (1-\delta^2)^{t}]\\
 &\leq \sum_{\psi\in M_{\mathcal{K},R}} \frac{\mathbb{E}|\langle\psi|U| +^n\rangle|^{2t}}{(1-\delta^2)^t}\\
 &\leq |\mathcal{K}|^R\frac{\sqrt{\mathbb{E}_{U\sim \zeta_{(n)}^{*2d}} |\langle+^n|U|+^n\rangle|^{2t}}}{(1-\delta^2)^{t}}.
 \end{split}
 \end{align}
 
 This is bounded by Theorem~\ref{prop:momentbound}.
 Eq.~\eqref{eq:complexitybound} in Theorem~\ref{thm:square-root-growth} follows from applying Eq.~\eqref{eq:momentboundintheorem}: If $d= 20k nt^2$ then
 \begin{align}
 \begin{split}
 \eqref{eq:highermarkovbound}&\leq |\mathcal{K}|^R \frac{\sqrt{t!2^{-nt}+2^{-nt}}+2^{-nt}}{(1-\delta^2)^{t}}\\
 &\leq 2^{-t(n-\log_2(t)+\log_2(1/1-\delta))+R\log_2(|\mathcal{K}|)}.
 \end{split}
 \end{align}
 In particular, we have with probability $1-2^{-\frac12 t(n-\log_2(t))}$ that 
 \begin{equation}
 R\geq \frac{t}{2\log_2(|\mathcal{K}|)}\left( n-\log_2(t)-\log_2\left(\frac{1}{1-\delta^2}\right)\right).
 \end{equation}
 This is almost Eq.~\eqref{eq:complexitybound} after inserting $t=\sqrt{d/20kn}$ except that we have bounded the complexity $C_{\mathcal{K},\delta}$ instead of $C_{\delta}$.
 For Eq.~\eqref{eq:complexitybound}, we will choose $\mathcal{K}_{\delta}$ to be a sufficiently close $\varepsilon$-net.
 Notice that an $\varepsilon$-net on $SU(4)$ exists of size $B^{\log(1/\varepsilon)}$ for a constant $B>0$.
 This follows from the existence of gapped gate sets~\cite{bourgain2012spectral} and the relation between approximation and gaps of averaging operators~\cite{harrow2002efficient,oszmaniec2021epsilon,varju2013random}.
 Moreover, a standard argument ensures that errors in the gates (in operators norm) add up linearly for quantum circuits~\cite{bennett1997strengths}.
 Suppose a state $|\psi\rangle$ can be approximated up to error $\delta$ in trace norm for $\delta\in (0,1/\sqrt{2})$ with circuits containing $R$ gates, then $|\langle\psi|V_R\cdots V_1|0^n\rangle|=\sqrt{1-\delta^2}$. 
 We know that there are gates $\tilde{V}_R,\ldots,\tilde{V}_1\in \mathcal{K}_{\delta}$ such that 
 \begin{equation}
 |\langle \psi|\tilde{V}_R\cdots \tilde{V}_1|0^n\rangle|^2\geq \sqrt{1-\delta^2}-\varepsilon R\geq \sqrt{1-2\delta^2},
 \end{equation}
 where the second inequality follows from choosing $\varepsilon=\delta^2/2d$.
 Consequently, $|\psi\rangle$ can be approximated up to an error of $\sqrt{2}\delta$ with gates from a $\delta^2/d$-net $\mathcal{K}_{\delta}$.
 Thus, $C_{\mathcal{K}_{\delta},\sqrt{2}\delta}\leq C_{\delta}$.
 We can therefore perform the above calculation~\eqref{eq:highermarkovbound} with $|\mathcal{K}_{\delta}|\leq B^{\log(d/\delta^2)}$.
 
For the proof of Eq.~\eqref{eq:linearboundintheorem}, we choose $t=2^n/2$.
For $d=20kn2^n$ we then find that Eq.~\eqref{eq:finalmomentboundintheorem} yields
\begin{align}
\begin{split}
\eqref{eq:highermarkovbound}&\leq |\mathcal{K}_{\delta}|^R \frac{\sqrt{2^{-n2^n/2+2^n/2}2^{n-1}!}+2^{-n2^{^n}/2}}{(1-\delta^2)^{2^n/2}}\\
&\leq 2^{-\frac14(1-2\log_2(1/(1-\delta^2)))2^n +R\log_2(|\mathcal{K}_{\delta}|)}.
\end{split}
\end{align}
 In particular, for constant $\delta<\sqrt{1-2^{-1/2}}$, we can lower bound the complexity 
 \begin{equation}\label{eq:complexityofexpdeepRQC}
 R\geq \frac{2^n}{8\log_2(|\mathcal{K}_{\delta}|)}\left(1-2\log_2\left(\frac{1}{1-\delta^2}\right)\right),
 \end{equation}
 with probability at least $1-2^{-\frac18(1-\log_2[1/(1-\delta^2)]) 2^n}$.

For any two unitaries $U$ and $U'$ we have that
\begin{multline}
||U-U'||^2_{\infty}=\max_{\psi}\langle\psi|(U-U')(U-U')^{\dagger}|\psi\rangle\\
=\max_{\psi}2(1-\mathrm{Re}\langle \psi|U^{\dagger}U'|\psi\rangle)\geq\max_{\psi} 2(1-|\langle \psi|U^{\dagger}U'|\psi\rangle|).
\end{multline}
In particular, if for a unitary $U$ all circuits $V$ of depth $R$ satisfy $|\langle +^n|V^{\dagger}U|+^n\rangle|\geq \sqrt{1-\delta^2}$ then $||V-U||^2_{\infty}\geq 2(1-\sqrt{1-\delta^2})=:\delta'$.
The rest of the proof now follows from the fact that a less-than-linear growth rate for most random quantum circuits leads to contradiction:
Consider a random quantum circuit of depth $d$.
We split the random quantum circuit of depth $d'=20kn2^n$ into $\lceil 20kn2^n/d \rceil$ many blocks, each with exactly $d$ gates except for the last, which has possibly less gates.
Each of these blocks have, with probability $1-e^{-\Omega( d/kn)}$, a circuit complexity of at least
\begin{align}\label{eq:boundoncomplexity}
\begin{split}
C_{\mathcal{K}_{\delta'},\frac{\delta' d}{\lfloor 20kn2^n \rfloor}}&\geq \lfloor 20kn2^n/d\rfloor^{-1}\left(\frac{2^n}{8\log_2(|\mathcal{K}_{\delta}|)}\left( 1-2\log_2\left(\frac{1}{1-\delta^2}\right)\right)-d\right)\\
&\geq \frac{d}{160kn\log_2(|\mathcal{K}_{\delta}|)}\left( 1-2\log_2\left(\frac{1}{1-\delta^2}\right)\right)-\frac{d^2}{20kn2^n}\\
&\geq \frac{d}{160kn\log_2(|\mathcal{K}_{\delta}|)}\left( 1-2\log_2\left(\frac{1}{1-\delta^2}\right)\right)-1,
\end{split}
\end{align}
where the last step is true under the assumption $d\leq \sqrt{20kn} 2^{n/2}$.
Indeed, suppose that a fraction larger than $e^{-O(d/kn)}$ does not satisfy this bound.
Then, these instances would add up to a circuit complexity less than $\frac{2^n}{8\log_2{|\mathcal{K}_{\delta}|}}\left(\frac14-2\log_2(1/(1-\delta^2))\right)$ with probability that cannot be upper bounded by $e^{-O(2^n)}$ contradicting Eq.~\eqref{eq:complexityofexpdeepRQC}.
\end{proof}

\section{Proof of Lemma~\ref{lemma:convergenceCNOTs}}\label{appendix:lemma1}
In this Section we prove Lemma~\ref{lemma:convergenceCNOTs}:
\begin{proof}[Proof of Lemma~\ref{lemma:convergenceCNOTs}]
For any probability measure $\nu$ on $\mathcal{Z}$ we define the averaging operator $T_{\nu}:L^2(\mathcal{Z}_n)\to L^2(\mathcal{Z}_n)$ via
\begin{equation}
(T_{\nu}f)(U):=\mathbb{E}_{V\sim \nu} f(V^{-1}U).
\end{equation}
The maximal eigenvalue of this operator is always $1$, realized on the subspace of constant functions.
By Ref.~\cite{diaconis1993comparison}, we have 
\begin{equation}
||T_\sigma-T_{\mu_{\mathcal{Z}}}||_{\infty}\leq 1-\frac{\eta}{d^2}\leq 1-\frac{1}{500n^5},
\end{equation}
where $\eta $ is the probability of drawing a specific generator given as an element of $\mathcal{Z}_2$ applied to a pair of qubits $(i,i+1)$ (therefore $\eta \geq 1/6n$) and $d$ is the number of generators necessary to generate any element of $\mathcal{Z}$ (therefore $d=9n^2$~\cite{bravyi2021hadamard}).
Thus, for any bounded function $f:\mathcal{Z}\to [-1,1]$ we have the following estimate:
\begin{align}
\begin{split}
|\mathbb{E}_{U\sim\sigma^{*k}}f(U)-\mathbb{E}_{U\sim\mu_{\mathcal{Z}}}f(U)|&\leq |T_{\sigma^{*k}}f(\mathbbm{1})-T_{\mu_{\mathcal{Z}}}f(\mathbbm{1})|\\
&\leq ||T_{\sigma^{*k}}f-T_{\mu_{\mathcal{Z}}}f||_2\\
&\leq \left(1-\frac{1}{500n^5}\right)^k||f||_2\\
&\leq \sqrt{|\mathcal{Z}|}\left(1-\frac{1}{500n^5}\right)^k.
\end{split}
\end{align}
Notice that the group of reversible circuits $\mathcal{Z}$ is contained in the Clifford group and thus~\cite{ozols2008clifford}
\begin{equation}
|\mathcal{Z}|\leq 2^{n^2+2n}\prod_j(4^j-1)\leq 2^{2n^2+2n},
\end{equation}
which completes the proof.
\end{proof}

\section{Outlook}
There are multiple open problems and avenues to continue this line of work:

\begin{itemize}
\item It would be interesting to relate the moment bound of $\zeta$-random quantum circuits to the more natural ensemble of quantum circuits with local gates drawn from the Haar measure $\mu_H$ on $SU(4)$.
Indeed, we have the operator inequality
\begin{equation}
    M((\mu_H)_{(n)},t)\leq M(\zeta_{(n)},t)
\end{equation}
 as $M(\zeta_{(n)},t)$ is a convex combination of orthogonal projectors and as such positive semidefinite. 
Moreover, its eigenvalue $1$ subspace contains $\mathrm{im}M((\mu_H)_{(n)},t)$.
Unfortunately, $(\bullet)^d$ is not an operator monotone and it is not clear if the inequality 
\begin{equation}
    M\left((\mu_H)_{(n)}^{*d},t\right)\leq M\left(\zeta_{(n)}^{*d},t\right)
\end{equation}
holds.
This leads to the fascinating conceptual question whether random quantum circuits with gates drawn from a subgroup can generate more random ensembles of states.
\item We have proven linear growth of quantum circuit complexity with an exponentially small error robustness.
The same method might be strong enough to imply linear growth for constant error robustness.
In fact, a simple conjecture about the Fourier support of Boolean functions would directly imply such a result in combination with the methods outlined in this work:
\begin{conjecture}\label{conjecture:boolean}
 Draw $2t$ bitstrings $x_1,\ldots,x_t, x'_1,\ldots,x'_t$ uniformly at random.
 Then, the support of $\hat{f}_{x_1,\ldots,x_t}^{x'_1,\ldots,x'_t}$, with $f_{x_1,\ldots,x_t}^{x'_1,\ldots,x'_t}$ as defined in~\eqref{eq:definefunctions}, contains at least $2^{n}/\mathrm{poly}(n)$ elements with probability $1-e^{-\Omega(t)}$ for all $t\leq e^{\Omega(n)}$.
\end{conjecture}

\item Another open problem is the spectral gap of random quantum circuits with gates drawn Haar randomly from $SU(4)$. 
Remarkably, the ideal auxiliary random walk on the diagonal subgroup has a spectral gap of exactly $2^{-n}$.
This eigenvalue appears for $t=2^{n}/2$, where the eigenstate corresponds to choosing $x_1,\ldots,x_t,x'_1,\ldots,x'_t$ such that $f_{x_1,\ldots,x_t}^{x'_1,\ldots,x'_t}$ is a parity check.
The same behavior might be true for random quantum circuits: An exponentially small unconditional gap but at the same time ``most'' eigenvalues are upper bounded by $1-1/\mathrm{poly}(n)$.

\item Bounds on eigenvalues of the moment operators for random quantum circuits tend to pick up large prefactors limiting the applicability in practical regimes. 
The original bound on the depth in Ref.~\cite{brandao_local_2016} has a prefactor that can be taken to be $4\times 10^7$ and the asymptotically improved version in Ref.~\cite{haferkamp2022random} even requires a depth of  $10^{13}$.
It would be desirable to reduce these constants.
Similarly, the factors polynomial in $n$ are likely artefacts of our proof technique.

\item Last, it would be interesting to prove monotonicity of robust circuit complexity under the application of random quantum gates.
Such a result would imply the bounds obtained in this work.
Moreover, monotonicity under free operations is key for the formulation of a resource theory of robust circuit (un)complexity~\cite{halpern2022resource}.

\end{itemize}
\section{Acknowledgments}
We want to thank Chi-Fang Chen, Yifan Jia, Richard Kueng and Ryan O'Donnell for fruitful discussions.
Moreover, we thank Nicole Yunger Halpern, Jens Eisert and Ryotaro Suzuki for detailed comments on this manuscript.
The author acknowledges funding from the Harvard Quantum Initiative.

\bibliographystyle{alpha}

\begin{thebibliography}{BCHJ{\etalchar{+}}21}

\bibitem[BaHH16]{brandao_local_2016}
F.~G. S.~L. Brand\~ao, A.~W. Harrow, and M.~Horodecki.
\newblock Local {random} {quantum} {circuits} are {approximate}
  {polynomial}-{designs}.
\newblock {\em Commun. Math. Phys.}, 346:397--434, 2016.

\bibitem[BBBV97]{bennett1997strengths}
C.~H. Bennett, E.~Bernstein, G.~Brassard, and U.~Vazirani.
\newblock Strengths and weaknesses of quantum computing.
\newblock {\em SIAM journal on Computing}, 26(5):1510--1523, 1997.

\bibitem[BCHJ{\etalchar{+}}21]{brandao_complexity_2019}
F.~G. S.~L. Brand{\~a}o, W.~Chemissany, N.~Hunter-Jones, R.~Kueng, and
  J.~Preskill.
\newblock Models of quantum complexity growth.
\newblock {\em PRX Quantum}, 2(3):030316, 2021.

\bibitem[BFG{\etalchar{+}}22]{bouland2022quantum}
A.~Bouland, B.~Fefferman, S.~Ghosh, U.~Vazirani, and Z.~Zhou.
\newblock Quantum pseudoentanglement.
\newblock {\em arXiv preprint arXiv:2211.00747}, 2022.

\bibitem[BG12]{bourgain2012spectral}
J.~Bourgain and A.~Gamburd.
\newblock A spectral gap theorem in su $(d) $.
\newblock {\em Journal of the European Mathematical Society}, 14(5):1455--1511,
  2012.

\bibitem[BHH16]{brandao2016efficient}
F.~G. S.~L. Brandao, A.~W. Harrow, and M.~Horodecki.
\newblock Efficient quantum pseudorandomness.
\newblock {\em Physical review letters}, 116(17):170502, 2016.

\bibitem[BM21]{bravyi2021hadamard}
S.~Bravyi and D.~Maslov.
\newblock Hadamard-free circuits expose the structure of the clifford group.
\newblock {\em IEEE Transactions on Information Theory}, 67(7):4546--4563,
  2021.

\bibitem[BS18]{brown2018second}
A.~R. Brown and L.~Susskind.
\newblock Second law of quantum complexity.
\newblock {\em Physical Review D}, 97(8):086015, 2018.

\bibitem[DSC93]{diaconis1993comparison}
P.~Diaconis and L.~Saloff-Coste.
\newblock Comparison techniques for random walk on finite groups.
\newblock {\em The Annals of Probability}, pages 2131--2156, 1993.

\bibitem[Haf22]{haferkamp2022random}
J.~Haferkamp.
\newblock Random quantum circuits are approximate unitary $t$-designs in depth
  ${O}(nt^{5+o(1)})$.
\newblock {\em Quantum}, 6:795, 2022.

\bibitem[HFK{\etalchar{+}}21]{haferkamp2021linear}
J.~Haferkamp, P.~Faist, N.~B.~T. Kothakonda, J.~Eisert, and N.~Yunger~Halpern.
\newblock Linear growth of quantum circuit complexity.
\newblock {\em Nature Physics}, 18:528--532, 2021.

\bibitem[HHJ21]{haferkamp2020improved}
J.~Haferkamp and N.~Hunter-Jones.
\newblock Improved spectral gaps for random quantum circuits: large local
  dimensions and all-to-all interactions.
\newblock {\em Phys. Rev. A}, 104:022417, 2021.

\bibitem[HJ19]{hunter2019unitary}
N.~Hunter-Jones.
\newblock Unitary designs from statistical mechanics in random quantum
  circuits.
\newblock {\em arXiv:1905.12053}, 2019.

\bibitem[HRC02]{harrow2002efficient}
A.~W. Harrow, B.~Recht, and I.~L. Chuang.
\newblock Efficient discrete approximations of quantum gates.
\newblock {\em Journal of Mathematical Physics}, 43(9):4445--4451, 2002.

\bibitem[JBS22]{jian2022linear}
Shao-Kai Jian, Gregory Bentsen, and Brian Swingle.
\newblock Linear growth of circuit complexity from brownian dynamics.
\newblock {\em arXiv preprint arXiv:2206.14205}, 2022.

\bibitem[JLS18]{ji2018pseudorandom}
Z.~Ji, Y.-K. Liu, and F.~Song.
\newblock Pseudorandom quantum states.
\newblock In {\em Advances in Cryptology--CRYPTO 2018: 38th Annual
  International Cryptology Conference, Santa Barbara, CA, USA, August 19--23,
  2018, Proceedings, Part III 38}, pages 126--152. Springer, 2018.

\bibitem[Kna88]{knabe1988energy}
S.~Knabe.
\newblock Energy gaps and elementary excitations for certain vbs-quantum
  antiferromagnets.
\newblock {\em Journal of statistical physics}, 52(3-4):627--638, 1988.

\bibitem[Li22]{li2022short}
Z.~Li.
\newblock Short proofs of linear growth of quantum circuit complexity.
\newblock {\em arXiv preprint arXiv:2205.05668}, 2022.

\bibitem[MR18]{maslov2018shorter}
Dmitri Maslov and Martin Roetteler.
\newblock Shorter stabilizer circuits via bruhat decomposition and quantum
  circuit transformations.
\newblock {\em IEEE Transactions on Information Theory}, 64:4729--4738, 2018.

\bibitem[Nac96]{nachtergaele1996spectral}
B.~Nachtergaele.
\newblock The spectral gap for some spin chains with discrete symmetry
  breaking.
\newblock {\em Communications in mathematical physics}, 175(3):565--606, 1996.

\bibitem[O'D14]{o2014analysis}
R.~O'Donnell.
\newblock {\em Analysis of boolean functions}.
\newblock Cambridge University Press, 2014.

\bibitem[OHHJ22]{oszmaniec2022saturation}
Micha{\l} Oszmaniec, Micha{\l} Horodecki, and Nicholas Hunter-Jones.
\newblock Saturation and recurrence of quantum complexity in random quantum
  circuits.
\newblock {\em arXiv preprint arXiv:2205.09734}, 2022.

\bibitem[Oli09]{oliveira2009convergence}
R.~I. Oliveira.
\newblock On the convergence to equilibrium of {Kac’s} random walk on
  matrices.
\newblock {\em The Annals of Applied Probability}, 19(3):1200--1231, 2009.

\bibitem[OSH21]{oszmaniec2021epsilon}
M.~Oszmaniec, A.~Sawicki, and M.~Horodecki.
\newblock Epsilon-nets, unitary designs and random quantum circuits.
\newblock {\em IEEE Transactions on Information Theory}, 2021.

\bibitem[Ozo08]{ozols2008clifford}
M.~Ozols.
\newblock Clifford group.
\newblock {\em Essays at University of Waterloo, Spring}, 2008.

\bibitem[RY17]{roberts2017chaos}
D.~A. Roberts and B.~Yoshida.
\newblock Chaos and complexity by design.
\newblock {\em JHEP}, 2017:121, 2017.

\bibitem[Var13]{varju2013random}
P.~Varj{\'u}.
\newblock Random walks in compact groups.
\newblock {\em Documenta Mathematica}, 18:1137--1175, 2013.

\bibitem[YKH{\etalchar{+}}22]{halpern2022resource}
N.~{Yunger Halpern}, N.~B.~T. Kothakonda, J.~Haferkamp, A.~Munson, J.~Eisert,
  and P.~Faist.
\newblock Resource theory of quantum uncomplexity.
\newblock {\em Physical Review A}, 106(6):062417, 2022.

\end{thebibliography}
\newcommand{\etalchar}[1]{$^{#1}$}

\end{document}